%% file: paper.tex
\documentclass{llncs}

\input{preamble}

\newtoggle{full}
\toggletrue{full}
\begin{document}

\title{Efficient Synthesis with Probabilistic Constraints}
\author{Samuel Drews \and Aws Albarghouthi \and Loris D'Antoni}
\institute{University of Wisconsin--Madison}

\maketitle
\begin{abstract}
We consider the problem of synthesizing a program
given a probabilistic specification of its desired behavior.
Specifically, we study the recent paradigm of \emph{distribution-guided inductive synthesis} (\alg), which iteratively calls a synthesizer
on finite sample sets from a given distribution.
% \alg ostensibly requires an exponential number of synthesizer invocations.
We make theoretical and algorithmic contributions:
\rone We prove the surprising result that \alg only requires a polynomial number of synthesizer calls in the size of the sample set, despite its ostensibly exponential behavior.
\rtwo We present a property-directed version of
\alg that further reduces the number of synthesizer calls,
drastically improving synthesis performance on a range of benchmarks.
%\keywords{First keyword  \and Second keyword \and Another keyword.}
\end{abstract}

\input{intro}

\input{prelims}

\input{theory}

\input{newopt}
\input{eval}
\input{related}

\paragraph{Acknowledgements}
We thank Shuchi Chawla, Yingyu Liang, Jerry Zhu,
the entire fairness reading group at UW-Madison,
and Nika Haghtalab
for all of the detailed discussions.
This material is based upon work supported by the National Science Foundation
under grant numbers 1566015, 1704117, and 1750965.

\bibliographystyle{splncs04}
\bibliography{biblio}

\iftoggle{full}{%
\appendix
\input{appex}

\input{apptoy}
\input{apptherm}

}{%
}

\end{document}

%% file: preamble.tex
\usepackage{prooftree}
\usepackage{amsmath}
\usepackage{amsfonts}
\usepackage{amssymb}
\usepackage{dsfont}
\usepackage{mathtools}
\usepackage{xspace}
\usepackage{xcolor}
\usepackage{etoolbox}
\usepackage[linesnumbered,noend]{algorithm2e}
\usepackage{pgfplots}
\usepgfplotslibrary{groupplots}
\usepackage{pgffor}
\usepackage{listings}
\usepackage{stmaryrd}
\usepackage{wrapfig}

\usepackage[T1]{fontenc}

\renewcommand{\paragraph}[1]{\vspace{.06in}\noindent\textbf{{#1}.}}
\newcommand{\abr}[1]{\textsc{\MakeLowercase{#1}}}

\newcommand{\rone}{(\emph{i})~}
\newcommand{\rtwo}{(\emph{ii})~}
\newcommand{\rthree}{(\emph{iii})~}

\newcommand{\alg}{\textsc{digits}\xspace}
\newcommand{\newalg}{$\tau$-\textsc{digits}\xspace}
\newcommand{\dist}{\mathds{D}}
\newcommand{\osynth}{\mathcal{O}_\textnormal{syn}}
\newcommand{\overif}{\mathcal{O}_\textnormal{ver}}
\newcommand{\oerror}{\mathcal{O}_\textnormal{err}}
\newcommand{\progs}{\ensuremath{\mathcal{P}}\xspace} % The set of programs / repair model
\newcommand{\sem}[1]{\llbracket #1\rrbracket}
\newcommand{\prog}{\ensuremath{P}\xspace} % a prog in progs
\newcommand{\initprog}{\ensuremath{\hat{P}}\xspace}

\newcommand{\post}{\textit{post}\xspace}
\newcommand{\error}{\text{Er}}
\newcommand{\hamm}{\textnormal{Hamming}}
\newcommand{\thresh}{\tau}
\newcommand{\dich}[2]{\operatorname{\Pi}_{#1}({#2})} % Realizable dichotimies in #1 of sample set #2
\newcommand{\growth}[2]{\operatorname{\hat{\Pi}}_{#1}({#2})} % Similarly, growth function
\renewcommand{\leq}{\leqslant}
\renewcommand{\geq}{\geqslant}

\newcommand{\pr}[2][]{%
    \ifstrempty{#1}{%
        \Pr \lbrack {#2} \rbrack
    }{%
        {\textstyle \Pr}_{#1} \lbrack {#2} \rbrack
    }%
}
\newcommand{\expec}[2][]{%
    \ifstrempty{#1}{%
        \textnormal{E} \lbrack {#2} \rbrack
    }{%
        \textnormal{E}_{#1} \lbrack {#2} \rbrack
    }%
}

%% file: intro.tex
\section{Introduction}\label{sec:intro}

% Program synthesis, the art of automatically generating programs that meet user intents,
% promises to increase the productivity of programmers
% and end-users of computing devices by automating tedious, error-prone, and time-consuming tasks.
% In most program synthesis problems, the user describes an intent using high-level specifications such as
% input-output examples or logical formulae, and the synthesizer magically produces a program that meets
% the user's intent.
%
% For a long time,
% the high theoretical complexities and
% the lack of practical applications prevented program synthesis from being used in practice.
% Recently, thanks to the advances in decision procedures
% and to the advent of the ``everyone wants to be a programmer'' era,
% the landscape has started to change and
% we have seen remarkable results in using synthesis to solve practical programs such as
% feedback generation in introductory programming assignments~\cite{SGS13,cav16},
Over the past few years, progress in automatic program synthesis
has touched many application domains, including
automating data wrangling and data extraction tasks~\cite{PolozovG15,RazaG17,WangGS16,Gulwani16,BarowyGHZ15,Gulwani11},
generating network configurations that meet user intents~\cite{popl17net,synet},
optimizing low-level code~\cite{SrinivasanR15,Schkufza0A16}, and more~\cite{Bastani0AL17,Gulwani14}.

The majority of the current work has focused on synthesis
under Boolean constraints. However, often times we require
the program to adhere to a probabilistic specification,
e.g., a controller that succeeds with a high probability,
a decision-making model operating over a probabilistic population model,
a randomized algorithm ensuring privacy, etc.
In this work, we are interested in
(1) investigating probabilistic synthesis from a theoretical perspective and (2) developing efficient algorithmic techniques to tackle this problem.

Our starting point is our recent framework for probabilistic synthesis
called \emph{distribution-guided inductive synthesis} (\alg)~\cite{digits}.
The \alg framework is analogous in nature to the \emph{guess-and-check} loop popularized by counterexample-guided approaches to synthesis and verification (\abr{CEGIS} and \abr{CEGAR}).
The key idea of the algorithm is reducing the probabilistic
synthesis problem to a non-probabilistic one that can be solved
using existing techniques, e.g., \abr{SAT} solvers.
This is performed using the following loop: (1) approximating the input probability distribution with a
finite sample set;
(2) synthesizing a program for
various possible
output assignments of the finite sample set;
and (3) invoking a probabilistic verifier to check if one of the synthesized programs
indeed adheres to the given specification.

 \alg has been shown to theoretically converge
to correct programs when they exist---thanks to learning-theory guarantees.
% However, the key practical limitation of \alg is the number of
% expensive invocations of the synthesizer $\osynth$ that is required for convergence,
% which in the worst-case is exponential in the number of samples.
The primary bottleneck of \alg is the number of expensive calls to the synthesizer,
which is ostensibly exponential in the size of the sample set.
Motivated by this observation, this paper makes theoretical,
algorithmic, and practical contributions:

\begin{itemize}
\item On the theoretical side, we present a detailed analysis of \alg and
prove that it only requires a polynomial number of invocations
of the synthesizer, explaining that the strong empirical performance of the algorithm is not merely due to the heuristics presented in \cite{digits} (Section~\ref{sec:theory}).

\item
On the algorithmic side, we develop an improved version of \alg that is
property-directed, in that it only invokes the synthesizer on instances that have a chance of resulting in a correct program, without sacrificing convergence. We call the new approach \newalg  (Section~\ref{sec:newopt}).

\item On the practical side, we implement \newalg
for sketch-based synthesis and demonstrate
its ability to converge significantly faster than \alg.
We apply our technique to a range of benchmarks, including
illustrative examples that elucidate our theoretical analysis,
probabilistic repair problems of unfair programs,
and probabilistic synthesis of controllers  (Section~\ref{sec:eval}).
\end{itemize}

%% file: prelims.tex
\section{An Overview of DIGITS}\label{sec:prelims}

In this section, we present the synthesis problem,
the \alg~\cite{digits} algorithm,
and fundamental background on learning theory.

\subsection{Probabilistic Synthesis Problem}

\paragraph{Program Model}
As discussed in~\cite{digits},
\alg searches through some (infinite) set of programs,
but it requires that the set of programs has \emph{finite VC dimension}
(we restate this condition in Section~\ref{sec:convergence}).
Here we describe one constructive way of obtaining such sets of programs
with finite VC dimension:
we will consider sets of programs defined as
\emph{program sketches}~\cite{solar08} in the simple grammar from~\cite{digits},
where a program is written in a loop-free language,
and ``holes'' defining the sketch replace some constant terminals in expressions.%
\footnote{In the case of loop-free program sketches as considered in our program model,
we can convert the input-output relation into a real arithmetic formula
that guaranteedly has finite VC dimension~\cite{goldberg95}.}
The syntax of the language is defined below:
\begin{align*}
    \prog	 \coloneqq~ V \gets E \mid \texttt{if } B \texttt{ then } \prog  \texttt{ else } \prog
    \mid \prog~\prog \mid \texttt{return } V
\end{align*}
Here, $\prog$ is a program,
$V$ is the set of variables appearing in $\prog$,
$E$ (resp. $B$) is the set of linear arithmetic (resp. Boolean) expressions over $V$
(where, again, constants in $E$ and $B$ can be replaced with holes),
and
$V \gets E$ is an assignment.
We assume a vector $\vec{v}_I$ of variables
 in  $V$ that are inputs to the program.
We also assume there is a single Boolean variable
$v_r \in V$ that is returned by the program.\footnote{
Restricting the output to Boolean is required by the algorithm;
other output types can be turned into Boolean by rewriting.
See, e.g., thermostat example in Section~\ref{sec:eval}.}
All variables are real-valued or Boolean.
%We always assume that programs are well-typed.
Given a vector of constant values $\vec{c}$,
where $|\vec{c}| = |\vec{v}_I|$,
we use $\prog(\vec{c})$ to denote
the result of executing $\prog$ on the input
$\vec{c}$.

In our setting, the inputs to a program are distributed according to some
\emph{joint probability distribution} $\dist$
over the variables $\vec{v}_I$.
Semantically, a program $\prog$ is denoted by a \emph{distribution
transformer} $\sem{\prog}$, whose input is a distribution
over values of $\vec{v}_I$ and whose output is a distribution over $\vec{v}_I$ and $v_r$.
% Formally, we think of the distribution
% $\pre$ as a \emph{probability space}
% $(\Omega, \mathcal{F}, \pr{\cdot})$: %,
% $\Omega$ is the set of possible
% assignments to $\vec{v}_I$,
% $\mathcal{F} \subseteq 2^\Omega$ is a set of \emph{events}, and
%  $\pr{\cdot}: \mathcal{F} \rightarrow [0,1]$
% denotes the probability of an event.

%We will be interested in two
%kinds of events:
%\begin{enumerate}
%  \item Given a Boolean expression $B$
%  over $\vec{v}_I$ and $v_r$, overloading notation,
%  a probability expression $\pr(B)$ denotes
%  $$\pr(\{\vec{c} \in \Omega \mid \exists r  \ldotp P(\vec{c}) = r \land
%  B[\vec{v}_I/\vec{c}, v_r/r] = \emph{true}\})$$
%  %
%  where the notation $B[\vec{x}/\vec{y}]$
%  denotes $B$ with all occurrences of $\vec{x}$ replaced by $\vec{y}$.
%  That is, $\pr(B)$ is the probability
%  of drawing a sample $\vec{c}$
%  from the precondition such that
%  the program $P$ returns a result satisfying
%  $B$.
%
%  \item Suppose we are given two programs
%  $P$ and $P'$ such that $\vec{v}_I$ and
%  $\vec{v}_I'$ are of the same length and type.
%  We will use $\pr(P \neq P')$
%  to denote:
%  $$\pr(\{\vec{c} \in \Omega \mid P(\vec{c}) \neq P'(\vec{c})\})$$
%  That is, $\pr(P \neq P')$, which we call the \emph{semantic distance}, is the probability
%  that the two programs return different
%  results on the same input.
%\end{enumerate}
%
%When we are talking about two different programs $P'$ and $P''$,
%we will use $\pr'(B)$ and $\pr''(B)$ to distinguish
%between the event $\pr(B)$ w.r.t program $P'$ versus $P''$.

%\paragraph{Probabilistic postconditions}
A program also has a   \emph{probabilistic
postcondition}, $\post$, defined as an inequality
over terms
of the form $\pr{B}$,
where $B$ is a Boolean expression over
$\vec{v}_I$  and $v_r$.
Specifically, a probabilistic postcondition
consists of Boolean combinations of the form $e > c$, where $c\in \mathds{R}$
and $e$ is an arithmetic expression over
terms of the form $\pr{B}$,
e.g., $\pr{B_1}/\pr{B_2} > 0.75$.

Given a triple $(\prog, \dist, \post)$,
we say that $P$ is \emph{correct} with respect
to $\dist$ and $\post$, denoted $\sem{\prog}(\dist) \models \post$,
\emph{iff} %the postcondition $\post$
$\post$ is true on the distribution $\sem{\prog}(\dist)$.

\begin{example}
Consider the set of intervals of the form $[0,a] \subseteq [0,1]$
and inputs $x$ uniformly distributed over $[0,1]$
(i.e.\ $\dist = \text{Uniform}[0,1]$).
We can write inclusion in the interval as a (C-style) program  (left)
and consider a postcondition stating that the interval must include
at least half the input probability mass (right):
\begin{center}
\vspace{-3mm}
\begin{minipage}{.45\textwidth}
\begin{lstlisting}[
    language=C++,
    basicstyle=\fontfamily{pcr}\selectfont\small,
]
if(0 <= x && x <= a) {
    return 1;
}
return 0;
\end{lstlisting}
\end{minipage}
\begin{minipage}{.45\textwidth}
\[
    \pr[x\sim\dist]{\prog(x) = 1} \geq 0.5
\]
\end{minipage}
\vspace{-3mm}
\end{center}
Let $\prog_c$ denote the interval program where $a$
is replaced by a constant $c \in [0, 1]$.
Observe that $\sem{\prog_c}(\dist)$
describes a joint distribution over $(x, v_r)$ pairs,
where $[0,c]\times\{1\}$ is assigned probability measure $c$
and $(c,1]\times\{0\}$ is assigned probability measure $1-c$.
Therefore, $\sem{\prog_c}(\dist) \models \post$
if and only if $c \in [0.5, 1]$.
\label{ex:prelim-ex}
\end{example}
%Recall Example~\ref{ex:motivatingex} and the problem of synthesizing a network switch that
%divides traffic between two next-hop switches.
%Packets arrive following a distribution $\dist$
%and the switch has to decide whether to forward each packet to switch $s_1$ or
%switch $s_2$.
%Assume we have a forwarding function $P$ that  violates the following property, which
%states that traffic is split almost equally among $s_1$ and $s_2$.
%\[
%\rand_{\dist} x.\ \pr{f(x)=s_1} > 0.45
%\wedge \pr{f(x)=s_2} > 0.45
%\]
%\sam{Using both $\rand$ and $\Pr$ is redundant
%unless $f$ is probabilistic itself
%(in which case we need to concretize some different notation).}
%We want a program $P'$ that satisfies the property and,
%since reconfiguring switches incurs in a cost, we require  $\pr[\dist]{P'\neq P}$ to be as small as possible.
%\alg solves this problem by sampling inputs from the packet distribution $\dist$
%and synthesizing a program for each possible way to output $s_1$ and $s_2$ on the inputs.

\paragraph{Synthesis Problem}
\alg outputs a program that is approximately ``similar'' to a given functional specification
and that meets a postcondition.
This functional specification is some input-output relation
which we quantitatively want to match as closely as possible:
specifically, we want to minimize the \emph{error} of the output program $\prog$
from the functional specification $\initprog$,
defined as $\error(\prog) \coloneqq \pr[x\sim\dist]{\prog(x)\neq\initprog(x)}$.
(Note that we represent the functional specification as a program.)
The postcondition is Boolean, and therefore we always want it to be true.
\alg is guaranteed to converge whenever the space of solutions satisfying the
postcondition is \emph{robust} under small perturbations.
The following definition captures this notion of robustness:

\begin{definition}[$\alpha$-Robust Programs]
Fix an input distribution $\dist$, a postcondition \post,
and a set of programs $\progs$.
For any $\prog \in \progs$ and any $\alpha > 0$,
denote the \emph{open $\alpha$-ball centered at \prog} as
$B_\alpha(\prog) = \{\prog' \in \progs \mid \pr[x \sim \dist]{P(x) \neq P'(x)} < \alpha\}$.
We say a program $\prog$ is \emph{$\alpha$-robust}
if $\forall \prog' \in B_\alpha(\prog) \ldotp \sem{\prog'}(\dist) \models \post$.
\end{definition}

We can now state the synthesis problem solved by \alg:
\begin{definition}[Synthesis Problem]
Given an input distribution $\dist$,
a set of programs \progs,
a postcondition \post,
a functional specification $\initprog \in \progs$, and
parameters $\alpha>0$ and $0< \varepsilon \leq \alpha$,
the  synthesis problem is to find a
program $\prog \in \progs$ such that
$\sem{\prog}(\dist) \models \post$ and
such that any other $\alpha$-robust $\prog'$
has $\error(\prog) \leq \error(\prog') + \varepsilon$.
\end{definition}

\subsection{A Naive DIGITS Algorithm}

Algorithm~\ref{alg:digits}
shows a simplified, naive version of \alg,
which employs a \emph{synthesize-then-verify} approach.
The idea of \alg is to utilize non-probabilistic
synthesis techniques to synthesize a set of programs,
and then apply a probabilistic verification step to check
if any of the synthesized programs is a solution.

\begin{wrapfigure}{r}{0.51\textwidth}
\begin{minipage}[t]{\linewidth}
\vspace{-7mm}
\begin{algorithm}[H]
\footnotesize
\SetAlgoLined\DontPrintSemicolon
\SetKwProg{procedure}{Procedure}{}{}
\procedure{$\alg(\initprog,\dist,\post,m)$}{
    $S \gets \{x \sim \dist \mid i \in [1,\ldots,m]\}$\;
    $\mathit{progs} \gets \emptyset$\;
    \ForEach{$f : S \rightarrow \{0,1\}$ \label{line:loop}} {
        $\prog \gets \osynth(\{(x,f(x)) \mid x \in S\})$\;
        \If{$\prog \neq  \bot$} {
            $\mathit{progs} \gets  \mathit{progs} \cup \{\prog\}$\;
        }
    }
    $\mathit{res} \gets \{\prog \in \mathit{progs} \mid \overif(\prog,\dist,\post)\}$\;
    \Return{$\textnormal{argmin}_{\prog \in \mathit{res}} \{\oerror(\prog)\}$}
}
\caption{Naive \alg \label{alg:digits}}
\end{algorithm}
\vspace{-5mm}
\end{minipage}
\end{wrapfigure}

Specifically, this ``Naive \alg'' begins by
sampling an appropriate number of inputs from the input distribution
and stores them in the set $S$.
Second, it iteratively explores each possible  
function $f$ that maps the input samples to a
Boolean and
invokes a synthesis oracle
to synthesize a program $\prog$ that implements $f$, i.e.\ that satisfies
the set of input--output examples in which each
input $x \in S$ is mapped to the output $f(x)$.
Naive \alg then finds which of the synthesized programs satisfy the postcondition (the set $\emph{res}$);
we assume that we have access to a probabilistic verifier $\overif$ to perform these computations.
Finally, the algorithm outputs the program in the set $\emph{res}$ that has the lowest error with respect to the functional specification,
once again assuming access to another oracle $\oerror$ that can measure the error.

Note that the number of such functions $f : S \rightarrow \{0,1\}$
is exponential in the size of $|S|$.
As a ``heuristic'' to improve performance,
the actual \alg algorithm as presented in~\cite{digits}
employs an incremental trie-based search,
which we describe (alongside our new algorithm, \newalg)
and analyze in Section~\ref{sec:theory}.
The naive version described here is, however, sufficient to discuss
the convergence properties of the full algorithm.

\subsection{Convergence Guarantees}\label{sec:convergence}
\alg is only guaranteed to converge when the program model $\progs$
has \emph{finite VC dimension}.%
\footnote{%
Recall that this is largely a ``free'' assumption
since, again, sketches in our loop-free grammar
guaranteedly have finite VC dimension.
}
Intuitively, the VC dimension captures the expressiveness
of the set of ($\{0,1\}$-valued) programs $\progs$.
Given a set of inputs $S$, we say that $\progs$ \emph{shatters} $S$ iff,
for every partition of $S$ into sets $S_0 \sqcup S_1$,
there exists a program $\prog \in \progs$ such that
\rone for every $x\in S_0$, $\prog(x)=0$, and
\rtwo for every $x\in S_1$, $\prog(x)=1$.

\begin{definition}[VC Dimension]
The \emph{VC dimension} of a set of programs $\progs$
is the largest integer $d$ such that there exists a set
of inputs $S$ with cardinality $d$ that is shattered by $\progs$.
\end{definition}

We define the function
$\textsc{VCcost}(\varepsilon, \delta, d)=\frac{1}{\varepsilon}
(4\log_2(\frac{2}{\delta})+8d\log_2(\frac{13}{\varepsilon}))$~\cite{blumer1989learnability},
which is used in the following theorem:
\begin{theorem}[Convergence]\label{thm:repairpac}
Assume that
there exist an $\alpha>0$ and program ${\prog^*}$
that is $\alpha$-robust w.r.t. $\dist$ and $\post$.
Let $d$ be the VC dimension of the set of programs $\progs$.
For all
 bounds $0<\varepsilon\leq\alpha$ and $\delta>0$,
for every function $\osynth$,
and for any $m \geq \textsc{VCcost}(\varepsilon, \delta, k)$,
with probability
$\geq 1-\delta$ we have that
\alg enumerates a program $\prog$ with
$\pr[x\sim\dist]{\prog^*(x) \neq \prog(x)} \leq \varepsilon$
and
$\sem{\prog}(\dist) \models \post$.
\end{theorem}

To reiterate, suppose $\prog^*$ is a correct program with small error $\error(\prog^*) = k$;
the convergence result follows two main points:
\rone $\prog^*$ must be \emph{$\alpha$-robust},
meaning every $\prog$ with $\pr[x\sim\dist]{\prog(x) \neq \prog^*(x)} < \alpha$
must also be correct, and therefore
\rtwo by synthesizing \emph{any} $\prog$ such that
$\pr[x\sim\dist]{\prog(x) \neq \prog^*(x)} \leq \varepsilon$
where $\varepsilon < \alpha$,
then $\prog$ is a correct program with error $\error(\prog)$ within $k \pm \varepsilon$.

\subsection{Understanding Convergence}
The importance of finite VC dimension is due to the fact
that the convergence statement borrows directly from
\emph{probably approximately correct (PAC) learning}.
We will briefly discuss a core detail of efficient PAC learning
that is relevant to understanding the convergence of \alg
(and, in turn, our analysis of \newalg in Section~\ref{sec:newopt}),
and refer the interested reader to Kearns and Vazirani's book~\cite{kearns94introduction}
for a complete overview.
Specifically, we consider the notion of an \emph{$\varepsilon$-net},
which establishes the approximate-definability of a target program
in terms of points in its input space.

\begin{definition}[$\varepsilon$-net]
Suppose $\prog \in \progs$ is a target program,
and points in its input domain $\mathcal{X}$ are distributed $x \sim \dist$.
For a fixed $\varepsilon \in [0, 1]$,
we say a set of points $S \subset \mathcal{X}$ is an \emph{$\varepsilon$-net}
for $\prog$ (with respect to $\progs$ and $\dist$) if
for every $\prog' \in \progs$ with $\pr[x \sim \dist]{\prog(x) \neq \prog'(x)} > \varepsilon$
there exists a witness $x \in S$ such that
$\prog(x) \neq \prog'(x)$.
\end{definition}
In other words, if $S$ is an $\varepsilon$-net for $\prog$,
and if $\prog'$ ``agrees'' with $\prog$ on all of $S$,
then $\prog$ and $\prog'$ can only differ by at most $\varepsilon$ probability mass.

Observe the relevance of $\varepsilon$-nets to the convergence of \alg:
the synthesis oracle is guaranteed not to ``fail''
by producing only programs $\varepsilon$-far from some $\varepsilon$-robust $\prog^*$
if the sample set happens to be an $\varepsilon$-net for $\prog^*$.
In fact, this observation is exactly the core of the PAC learning argument:
having an $\varepsilon$-net exactly guarantees the approximate learnability.

A remarkable result of computational learning theory is that
whenever $\progs$ has finite VC dimension,
the probability that $m$ random samples fail to yield an $\varepsilon$-net
becomes diminishingly small as $m$ increases.
Indeed, the given \textsc{VCcost} function used in Theorem~\ref{thm:repairpac}
is a dual form of this latter result---%
that polynomially many samples are sufficient
to form an $\varepsilon$-net with high probability.

%% file: theory.tex
\section{The Efficiency of Trie-Based Search}\label{sec:theory}
After providing details on the search strategy employed by
\alg, we present our theoretical result on the polynomial bound on the number of synthesis
queries that \alg requires.

\subsection{The Trie-Based Search Strategy of DIGITS}
Naive \alg, as presented in Algorithm~\ref{alg:digits}, performs a very unstructured, exponential search
over the output labelings of the sampled inputs---i.e., the possible Boolean functions $f$ in Algorithm~\ref{alg:digits}.
In our original paper~\cite{digits} we present a ``heuristic'' implementation strategy
that incrementally explores the set of possible output labelings
using a  trie data structure.
%this incremental version exploits
%an \emph{implicit} structure of the set of possible programs.
In this section, we study the complexity of this technique through the lens  of computational learning theory
and discover the surprising result that \alg requires
a polynomial number of calls to the synthesizer in the
size of the sample set!
Our improved search algorithm (Section~\ref{sec:newopt}) inherits these results.

For the remainder of this paper,
we use \alg to refer to this incremental version.
A full description is necessary for our analysis:
Figure~\ref{fig:newalg} (non-framed rules only) consists of a collection of guarded rules
describing the construction of the trie used by \alg
to incrementally explore the set of possible output labelings.
% Actually, Figure~\ref{fig:newalg} contains more than just this incremental algorithm---%
Our improved version, \newalg (presented in Section~\ref{sec:newopt}),
corresponds to the addition of the framed parts,
but without them, the rules describe \alg.

\input{algfig} % unclear where this should go

Nodes in the trie represent partial output labelings---i.e., functions
$f$ assigning Boolean values to only some of the samples in $S=\{x_1,\ldots,x_m\}$.
Each node is identified by a binary string $\sigma=b_1\cdots b_k$  ($k$ can be smaller than $m$)
denoting the path to the node from the root.
The string $\sigma$ also describes the partial output-labeling function $f$ corresponding to
the node---i.e., if the $i$-th bit $b_i$ is set to 1, then $f(x_i)=true$.
The set $\mathit{explored}$ represents the nodes in the trie built thus far;
for each new node, the algorithm synthesizes a program consistent with
the corresponding partial output function (``Explore'' rules).
The variable $\mathit{depth}$ controls the incremental aspect of the search
and represents the maximum length of any $\sigma$ in $\mathit{explored}$;
it is incremented whenever all nodes up to that depth
have been explored (the ``Deepen'' rule).
The crucial part of the algorithm is that, if no program can be synthesized for the
partial output function of a node identified by $\sigma$, the
algorithm does not need to issue further synthesis queries for the descendants of $\sigma$.

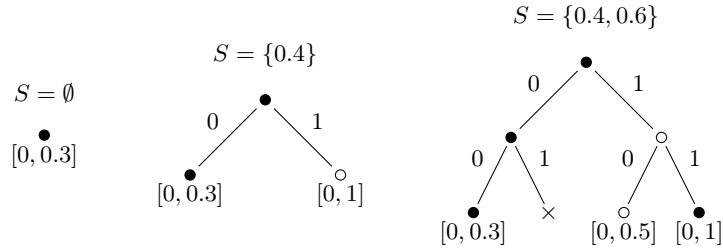
\begin{figure}[t]
\centering
\begin{minipage}{.15\textwidth}
\begin{tikzpicture}[x=1cm, y=1cm]
    \node (root) at (0,0) {};
    \draw[fill] (root) circle (.5ex);
    \node at (root) [above=1em] {$S = \emptyset$};
    \node[below] at (root) {$[0,0.3]$};
\end{tikzpicture}
\end{minipage}
\begin{minipage}{.3\textwidth}
\begin{tikzpicture}[x=1cm, y=1cm]
    \node (root) at (0,0) {};
    \draw[fill] (root) circle (.5ex);
    \node at (root) [above=1em] {$S = \{0.4\}$};

    \node (c0) at (-1, -1) {};
    \node (c1) at (1, -1) {};
    \draw[fill] (c0) circle (.5ex);
    \draw (c1) circle (.5ex);
    \node[below] at (c0) {$[0,0.3]$};
    \node[below] at (c1) {$[0,1]$};

    \draw (root) -- node[above left] {0} ++(c0);
    \draw (root) -- node[above right] {1} ++(c1);
\end{tikzpicture}
\end{minipage}
\begin{minipage}{.3\textwidth}
\begin{tikzpicture}[x=1cm, y=1cm]
    \node (root) at (0,0) {};
    \draw[fill] (root) circle (.5ex);
    \node at (root) [above=1em] {$S = \{0.4, 0.6\}$};

    \node (c0) at (-1, -1) {};
    \node (c1) at (1, -1) {};
    \draw[fill] (c0) circle (.5ex);
    \draw (c1) circle (.5ex);

    \draw (root) -- node[above left] {0} ++(c0);
    \draw (root) -- node[above right] {1} ++(c1);

    \node (c00) at (-1.5,-2) {};
    \node (c01) at (-0.5,-2) {};
    \node (c10) at (0.5,-2) {};
    \node (c11) at (1.5,-2) {};
    \draw[fill] (c00) circle (.5ex);
    \node at (c01) {$\times$};
    \draw (c10) circle (.5ex);
    \draw[fill] (c11) circle (.5ex);
    \node[below] at (c00) {$[0,0.3]$};
    \node[below] at (c10) {$[0,0.5]$};
    \node[below] at (c11) {$[0,1]$};
    \draw (c0) -- node[above left] {0} ++(c00);
    \draw (c0) -- node[above right] {1} ++(c01);
    \draw (c1) -- node[above left] {0} ++(c10);
    \draw (c1) -- node[above right] {1} ++(c11);
\end{tikzpicture}
\end{minipage}
\caption{Example execution of incremental \alg
on interval programs, starting from $[0,0.3]$.
Hollow circles denote calls to $\osynth$ that yield new programs;
the cross denotes a call to $\osynth$ that returns $\bot$.}
\label{fig:ex}
\end{figure}

Figure~\ref{fig:ex} shows how \alg builds a trie
for an example run on the interval programs from Example~\ref{ex:prelim-ex},
where we suppose we begin with an incorrect program
describing the interval $[0,0.3]$.
Initially, we set the root program to $[0,0.3]$ (left figure).
The ``Deepen'' rule applies,
so a sample is added to the set of samples---suppose it's $0.4$.
``Explore'' rules are then applied twice to build the children of the root:
the child following the 0 branch needs to map $0.4 \mapsto 0$,
which $[0,0.3]$ already does, thus it is propagated to that child
without asking $\osynth$ to perform a synthesis query.
For the child following 1, we instead make a synthesis query,
using the oracle $\osynth$, for any value of $a$ such that $[0,a]$
maps $0.4 \mapsto 1$---suppose it returns the solution $a=1$,
and we associate $[0,1]$ with this node.
At this point we have exhausted depth 1 (middle figure),
so ``Deepen'' once again applies, perhaps adding $0.6$ to the sample set.
At this depth (right figure), only two calls to $\osynth$ are made:
in the case of the call at $\sigma=01$,
there is no value of $a$ that causes both $0.4 \mapsto 0$ and $0.6 \mapsto 1$,
so $\osynth$ returns $\bot$, and we do not try to explore any children of this node
in the future.
The algorithm continues in this manner until a stopping condition is reached---e.g.,
enough samples are enumerated.

\subsection{Polynomial Bound on the Number of Synthesis Queries}

We observed in~\cite{digits} that the trie-based exploration seems to be efficient in practice,
despite potential exponential growth of the number of explored nodes in the trie as the depth of the search increases.
The convergence analysis of \alg relies on the finite VC dimension of the program model,
but VC dimension itself is just a summary of the \emph{growth function},
a function that describes a notion of complexity of the set of programs in question.
We will see that the growth function much more precisely describes the behavior of the
trie-based search;
we will then use a classic result from computational learning theory
to derive better bounds on the performance of the search.
We define the growth function below,
adapting the presentation from~\cite{kearns94introduction}.
% note that the computational learning literature defines these notions in terms of sets of subsets,
% whereas we will define them (equivalently) in terms of sets of $\{0,1\}$-valued functions/programs:
\begin{definition}[Realizable Dichotomies]
We are given a set $\progs$ of programs representing functions from $\mathcal{X} \rightarrow \{0,1\}$
and a (finite) set of inputs $S \subset \mathcal{X}$.
We call any $f : S \rightarrow \{0,1\}$ a \emph{dichotomy} of $S$;
if there exists a program $\prog \in \progs$ that extends $f$ to its full domain $\mathcal{X}$,
we call $f$ a \emph{realizable dichotomy} in $\progs$.
We denote the set of realizable dichotomies as
\[
    \dich{\progs}{S} \coloneqq
    \{f : S \rightarrow \{0,1\} \mid
    \exists \prog \in \progs \ldotp \forall x \in S \ldotp
    \prog(x) = f(x)\}.
\]
\end{definition}
Observe that for any (infinite) set $\progs$ and any finite set $S$
that $1 \leq \lvert\dich{\progs}{S}\rvert \leq 2^{|S|}$.
We define the growth function in terms of the realizable dichotomies:
\begin{definition}[Growth Function]
The \emph{growth function} is the maximal number of realizable dichotomies
as a function of the number of samples, denoted
\[
    \growth{\progs}{m} \coloneqq
    \max_{\substack{S \subset \mathcal{X} :\\ |S| = m}} \{ \lvert\dich{\progs}{S}\rvert \}.
\]
\end{definition}
Observe that $\progs$ has VC dimension $d$ if and only if
$d$ is the largest integer satisfying $\growth{\progs}{d} = 2^d$
(and infinite VC dimension when $\growth{\progs}{m}$ is identically $2^m$)---%
in fact, VC dimension is often defined using this characterization.

\begin{example}
Consider the set of intervals of the form $[0,a]$
as in Example~\ref{ex:prelim-ex} and Figure~\ref{fig:ex}.
For the set of two points $S = \{0.4, 0.6\}$, we have that $\lvert\dich{[0,a]}{S}\rvert = 3$,
since, by example:
$a = 0.5$ accepts $0.4$ but not $0.6$,
$a = 0.3$ accepts neither, and $a = 1$ accepts both,
thus these three dichotomies are realizable;
however, no interval with $0$ as a left endpoint can accept $0.6$ and not $0.4$,
thus this dichotomy is not realizable.
In fact, for any (finite) set $S \subset [0, 1]$,
we have that $\lvert\dich{[0,a]}{S}\rvert = |S| + 1$;
we then have that $\growth{[0,a]}{m} = m + 1$.
\end{example}

When \alg terminates having used a sample set $S$,
it has considered all the dichotomies of $S$:
the programs it has enumerated exactly correspond to extensions
of the realizable dichotomies $\dich{\progs}{S}$.
The trie-based exploration
is effectively trying to minimize the number of $\osynth$ queries performed on non-realizable ones,
but doing so without explicit knowledge of the full functional behavior of programs in $\progs$.
In fact, it manages to stay relatively close to performing queries
only on the realizable dichotomies:
\begin{lemma}\label{lem:algdich}
\alg performs at most $|S| \lvert\dich{\progs}{S}\rvert$ synthesis oracle queries.
More precisely, let $S = \{x_1,\ldots,x_m\}$ be indexed
by the depth at which each sample was added:
the exact number of synthesis queries is
$\sum_{\ell = 1}^m \lvert\dich{\progs}{\{x_1,\ldots,x_{\ell-1}\}}\rvert$.
\end{lemma}
\begin{proof}
Let $T_d$ denote the total number of queries performed once depth $d$ is completed.
We perform no queries for the root,%
\footnote{We assume the functional specification itself
is some $\initprog \in \progs$ and thus can be used---%
the alternative is a trivial synthesis query on an empty set of constraints.}
thus $T_0 = 0$.
Upon completing depth $d-1$, the realizable dichotomies of
$\{x_1, \ldots, x_{d-1}\}$ exactly specify the nodes
whose children will be explored at depth $d$.
For each such node, one child is skipped due to solution propagation,
while an oracle query is performed on the other,
thus $T_d = T_{d-1} + \lvert\dich{\progs}{\{x_1,\ldots,x_{d-1}\}}\rvert$.
Lastly, $|\dich{\progs}{S}|$ cannot decrease by adding elements to $S$, so we have that
$T_m = \sum_{\ell=1}^m \lvert\dich{\progs}{\{x_1,\ldots,x_{\ell-1}\}}\rvert
\leq \sum_{\ell=1}^m \lvert\dich{\progs}{S}\rvert
\leq |S| \lvert\dich{\progs}{S}\rvert$.
\qed
\end{proof}

%\begin{example}
%Again, consider the intervals $[0,a]$ as in Figure~\ref{fig:ex},
%and note that in the example run,
%for the full set $\{0.4, 0.6, 0.2\}$ there are 4 realizable dichotomies;
%6 synthesis queries are performed,
%which (informally) is not \emph{too many} more.
%For any run of the inductive algorithm on this set of programs,
%$\frac{1}{2}|S|(|S|+1)$ synthesis queries will be performed,
%which is only a multiplicative factor of $\frac{1}{2}|S|$ larger than
%the number of realizable dichotomies (recall $\lvert\dich{[0,a]}{S}\rvert = |S| + 1$)
%and only a factor of $\frac{1}{2}$ smaller than the looser bound
%given by Lemma~\ref{lem:algdich}: $|S| (|S| + 1)$.
%\end{example}

Connecting \alg to the realizable dichotomies and, in turn, the growth function
allows us to employ a remarkable result from computational learning theory,
stating that the growth function for any set exhibits one of two asymptotic behaviors:
it is either \emph{identically} $2^m$ (infinite VC dimension)
or dominated by a polynomial!
%Vapnik and Chervonenkis prove this, themselves~\cite{vapnik71},
This is commonly called the Sauer-Shelah Lemma~\cite{sauer72,shelah72}:
\begin{lemma}[Sauer-Shelah]\label{lem:sauer}
If $\progs$ has finite VC dimension $d$,
then for all $m \geq d$, $\growth{\progs}{m} \leq \left(\frac{em}{d}\right)^d$;
i.e.\ $\growth{\progs}{m} = O(m^d)$.
\end{lemma}

Combining our lemma with this famous one yields a surprising result---%
that for a fixed set of programs $\progs$ with finite VC dimension,
the number of oracle queries performed by \alg
\emph{is guaranteedly polynomial} in the depth of the search,
where the degree of the polynomial is determined by the VC dimension:
\begin{theorem}\label{thm:main}
If $\progs$ has VC dimension $d$,
then \alg performs $O(m^{d+1})$ synthesis-oracle queries.
\end{theorem}

In short, the reason an execution of \alg \emph{seems} to enumerate
a sub-exponential number of programs (as a function of the depth of the search)
is because it literally must be polynomial.
Furthermore, the algorithm performs oracle queries
on \emph{nearly} only those polynomially-many realizable dichotomies.

\begin{example}
A \alg run on the $[0,a]$ programs as in Figure~\ref{fig:ex}
using a sample set of size $m$
will perform $O(m^2)$ oracle queries,
since the VC dimension of these intervals is $1$.
%(since $\growth{[0,a]}{m} = m + 1$ and $m + 1 = 2^m$ for $m = 1$ but no larger).
%Note that in the previous example, we actually saw that the number of queries
%is exactly $\frac{1}{2}m(m+1)$.
(In fact, every run of the algorithm on these programs
will perform exactly $\frac{1}{2}m(m+1)$ many queries.)
\end{example}

%% file: algfig.tex
\begin{figure}[t]
\centering
\newcommand{\highlight}[1]{\framebox{$#1$}}
\begin{prooftree}
    \justifies
    \mathit{explored} \gets \{\epsilon\}
    \quad \prog_\epsilon \gets \initprog
    \quad \mathit{depth} \gets 0
    \quad \mathit{best} \gets \bot
    \using \text{Initialize}
\end{prooftree}
\vspace{1em}

\begin{prooftree}
    \begin{array}{c}
    \forall \sigma \in \mathit{explored} \ldotp \forall b \in \{0,1\} \ldotp \\
    (\prog_\sigma \neq \bot \land |\sigma b| \leq \mathit{depth}
    \highlight{\land \mathit{unblocked}(\sigma b)})
    \Rightarrow \sigma b \in \mathit{explored}
    \end{array}
    \justifies
    \mathit{sample}_{\mathit{depth}+1} \sim \dist
    \quad \mathit{depth} \gets \mathit{depth} + 1
    \using \text{Deepen}
\end{prooftree}
\vspace{1em}

\begin{prooftree}
    \begin{array}{c}
    \sigma \in \mathit{explored}
    \quad \prog_\sigma \neq \bot
    \quad b \in \{0,1\} \\
    \sigma b \not \in \mathit{explored}
    \quad |\sigma b| \leq \mathit{depth}
    \quad \highlight{\mathit{unblocked}(\sigma b)}
    \end{array}
    \justifies
    \begin{array}{c}
    \prog_{\sigma b} \gets
        \osynth(\{(\mathit{sample}_{i+1},\sigma b (i)) : 0 \leq i < |\sigma b |\}) \\
    \quad \mathit{explored} \gets \mathit{explored} \cup \{\sigma b\}
    \end{array}
    \using \text{Explore (Synthesis Query)}
\end{prooftree}
\vspace{1em}

\begin{prooftree}
    \begin{array}{c}
    \sigma \in \mathit{explored}
    \quad \prog_\sigma \neq \bot
    \quad b \in \{0,1\}
    \quad \sigma b \not \in \mathit{explored} \\
    |\sigma b| \leq \mathit{depth}
    \quad \highlight{\mathit{unblocked}(\sigma b)}
    \quad \prog_\sigma(\mathit{sample}_{|\sigma b|}) = b
    \end{array}
    \justifies
    \prog_{\sigma b} \gets \prog_\sigma
    \quad \mathit{explored} \gets \mathit{explored} \cup \{\sigma b\}
    \using \text{Explore (Solution Propagation)}
\end{prooftree}
\vspace{1em}

\begin{prooftree}
    \sigma^* = \textnormal{argmin}_\sigma \{\oerror(\prog_\sigma) \mid
        \sigma \in \mathit{explored} \land \prog_\sigma \neq \bot
        \land \overif(\prog_\sigma) = \textnormal{true} \}
    \justifies
    \mathit{best} \gets \prog_{\sigma^*}
    \using \text{Best}
\end{prooftree}
\vspace{1em}

\highlight{
$$\text{where } \mathit{unblocked}(\sigma) \coloneqq
|\{i : 0 \leq i < |\sigma| \land \sigma(i) \neq \initprog(\mathit{sample}_{i+1})\}|
\leq \thresh \cdot \mathit{depth}$$
}

\caption{Full \alg description and our new extension, \newalg, shown in boxes.}
\label{fig:newalg}
\end{figure}

%% file: newopt.tex
\section{Property-Directed $\tau$-DIGITS}\label{sec:newopt}

\alg has better convergence guarantees when it operates on larger sets of sampled inputs.
In this section, we describe a new optimization of \alg that reduces the number of synthesis queries
performed by the algorithm so that it more quickly reaches higher depths in the trie, and
thus allows to scale to larger samples sets.
This optimized \alg, called \newalg, is shown in Figure~\ref{fig:newalg} as the set of all the rules of \alg plus the framed elements.
The high-level idea is to skip synthesis queries that are (quantifiably) unlikely to result
in optimal solutions.
For example, if the functional specification $\initprog$ maps every sampled input in $S$ to 0,
then the synthesis query on the mapping of every element of $S$ to 1
becomes increasingly likely to result in programs that have maximal distance from $\initprog$
as the size of $S$ increases;
hence the algorithm could probably avoid performing that query.
In the following, we make use of the concept of \emph{Hamming distance}
between pairs of programs:
\begin{definition}[Hamming Distance]
For any finite set of inputs $S$ and any two programs $\prog_1, \prog_2$,
we denote
$ \hamm_S(\prog_1, \prog_2) \coloneqq |\{x \in S \mid \prog_1(x) \neq \prog_2(x)\}| $
(we will also allow any $\{0,1\}$-valued string
to be an argument of $\hamm_S$).
\end{definition}

\subsection{Algorithm Description}
%\sam{I like still using $k$, redundantly alonside $\error$,
%since especially in the binomial distribution details,
%the $\error$ notation starts to be very cumbersome.
%Let me know if you disagree.}

Fix the given functional specification $\initprog$
and suppose that there exists an $\varepsilon$-robust solution $\prog^*$
with (nearly) minimal error $k = \error(\prog^*) \coloneqq \pr[x\sim\dist]{\initprog(x) \neq \prog^*(x)}$;
%\loris{can't be minimal if it has  $\epsilon$-ball}
%\sam{it can be minimal from the set that have the $\varepsilon$-ball,
%except that I defined them as open balls, so it has to be the infimum.
%This latter detail is discussed in the original paper,
%and I personally don't think it's very interesting.
%If only the former part needs to be clarified,
%that can happen in the next pass}
we would be happy to find \emph{any} program $\prog$ in $\prog^*$'s $\varepsilon$-ball.
Suppose we angelically know $k$ a priori,
and we thus restrict our search (for each depth $m$) only to constraint strings
(i.e.\ $\sigma$ in Figure~\ref{fig:newalg})
that have Hamming distance not much larger than $km$.

%\aws{last sentence doesn't type check; $k$ is [0,1], hamming integer}
%\loris{$\epsilon$-robust, $\epsilon$-robust: change to $\alpha$-robust, $\alpha$-robust}

To be specific, we first fix some threshold $\thresh \in (k,1]$.
Intuitively, the optimization corresponds to modifying \alg to consider only paths $\sigma$ through the trie such that $\hamm_S(\initprog, \sigma) \leq \thresh |S|$.
This is performed using the \emph{unblocked} function in Figure~\ref{fig:newalg}.
Since we are ignoring certain paths through the trie, we need to ask:
\emph{How much does this decrease the probability of the algorithm succeeding?}---%
It depends on the tightness of the threshold, which we address in Section~\ref{sec:failureprob}.
In Section~\ref{subsec:adaptive}, we discuss how to
adaptively modify the threshold $\thresh$ as \newalg is executing,
which is useful when a good $\tau$ is unknown a priori.

\subsection{Analyzing Failure Probability with Thresholding}
\label{sec:failureprob}

Using \newalg, the choice of $\tau$ will affect both
\rone how many synthesis queries are performed, and
\rtwo the likelihood that we \emph{miss} optimal solutions;
in this section we explore the latter point.\footnote{The former point is a difficult combinatorial question
that to our knowledge has no precedent in the computational learning literature, and so we leave it as future work.}
Interestingly, we will see that all of the analysis is dependent only
on parameters directly related to the threshold;
notably, none of this analysis is dependent on the complexity of $\progs$
(i.e.\ its VC dimension).

If we really want to learn (something close to) a program $\prog^*$,
then we should use a value of the threshold $\tau$ such that
$\pr[S\sim\dist^m]{\hamm_S(\initprog,\prog^*) \leq \thresh m}$ is large---%
to do so requires knowledge of the distribution of $\hamm_S(\initprog,\prog^*)$.
%Observe that the relevant expected value nicely decomposes:
%\[
%\expec[S\sim\dist^m]{\hamm_S(\initprog,\prog^*)}
%= m \pr[x\sim\dist]{\initprog(x) \neq \prog^*(x)}
%\]
%so we should choose $\tau \geq \pr[x\sim\dist]{\initprog(x) \neq \prog^*(x)}$
%appropriately large, based on the distribution of $\hamm_S(\initprog,\prog^*)$.
Recall the \emph{binomial distribution}:
for parameters $(n,p)$, it describes the number of successes
in $n$-many trials of an experiment that has success probability $p$.
\begin{claim}
Fix $\prog$ and let $k = \pr[x\sim\dist]{\initprog(x) \neq \prog(x)}$.
If $S$ is sampled from $\dist^m$, then $\hamm_S(\initprog, \prog)$ is
binomially distributed with parameters $(m, k)$.
\end{claim}
Next, we will use our knowledge of this distribution
to reason about the \emph{failure probability},
i.e.\ that \newalg does not preserve the convergence result of \alg.

The simplest argument we can make is a union-bound style argument:
the thresholded algorithm can ``fail'' by
\rone failing to sample an $\varepsilon$-net, or otherwise
\rtwo sampling a set on which the optimal solution has a Hamming distance
that is not representative of its actual distance.
We provide the quantification of this failure probability in the following theorem:
\begin{theorem}\label{thm:thresherror}
Let $\prog^*$ be a target $\varepsilon$-robust program
with $k = \pr[x\sim\dist]{\initprog(x) \neq \prog^*(x)}$,
and let $\delta$ be the probability that $m$ samples
do not form an $\varepsilon$-net for $\prog^*$.
If we run the \newalg with $\thresh \in (k,1]$,
then the failure probability is at most $\delta + \pr{X > \thresh m}$
where $X \sim \textnormal{Binomial}(m, k)$.
\end{theorem}
In other words, we can use tail probabilities of the binomial distribution
to bound the probability that the threshold
causes us to ``miss'' a desirable program we otherwise would have enumerated.
Explicitly, we have the following corollary:
\begin{corollary}\label{cor:error}
\newalg increases failure probability
(relative to \alg)  by at most
$
\pr{X > \thresh m} =
\sum_{i = \lfloor \thresh m \rfloor + 1}^m {m \choose i} k^i (1-k)^{m-i}.
$
\end{corollary}
Informally, when $m$ is \emph{not too small},
$k$ is \emph{not too large},
and $\thresh$ is \emph{reasonably forgiving},
these tail probabilities can be quite small.
We can even analyze the asymptotic behavior by using any existing upper bounds
on the binomial distribution's tail probabilities---%
importantly, the additional error diminishes exponentially as $m$ increases,
dependent on the size of $\thresh$ relative to $k$.
\begin{corollary}\label{cor:asymp}
\newalg increases failure probability by at most
$e^{-2m(\thresh-k)^2}$.%
\footnote{
A more precise (though less convenient) bound is
$e^{-m(\thresh \ln \frac{\thresh}{k} + (1-\thresh) \ln \frac{1-\thresh}{1-k})}$.
}
\end{corollary}

\begin{example}
Suppose $m = 100$, $k = 0.1$, and $\thresh = 0.2$.
Then the extra failure probability term in Theorem~\ref{thm:thresherror}
is less than $0.001$.
\end{example}

As stated at the beginning of this subsection,
the balancing act is to choose $\thresh$
\rone small enough so that the algorithm is still fast for large $m$, yet
\rtwo large enough so that the algorithm is still likely to learn the desired programs.
The further challenge is to relax our initial strong assumption
that we know the optimal $k$ a priori when determining $\tau$,
which we address in the following subsection.

\subsection{Adaptive Threshold}\label{subsec:adaptive}

Of course, we do not have the angelic knowledge that lets us pick
an ideal threshold $\thresh$;
the only absolutely sound choice we can make is the trivial $\thresh = 1$.
Fortunately, we can begin with this choice of $\thresh$
and \emph{adaptively} refine it as the search progresses.
Specifically, every time we encounter a correct program $\prog$
such that $k = \error(\prog)$,
we can refine $\thresh$ to reflect our newfound knowledge that
``the best solution has distance of at most $k$.''

We refer to this refinement as \emph{adaptive \newalg}.
The modification involves the addition of the following rule to Figure~\ref{fig:newalg}:

\begin{center}
\begin{prooftree}
    \mathit{best} \neq \bot
    \justifies
    \thresh \gets g(\oerror(\mathit{best}))
    \using \text{Refine Threshold (for some $g : [0,1] \rightarrow [0,1]$)}
\end{prooftree}
\end{center}

We can use any (non-decreasing) function $g$ to update the threshold
$\thresh \gets g(k)$.
The simplest choice would be the identity function (which we use in our experiments),
although one could use a looser function so as not to over-prune the search.
If we choose functions of the form $g(k) = k + b$,
then Corollary~\ref{cor:asymp} allows us to make (slightly weak) claims of the following form:
\begin{claim}
Suppose the adaptive algorithm completes a search of up to depth $m$
yielding a best solution with error $k$
(so we have the final threshold value $\thresh = k + b$).
Suppose also that $\prog^*$ is an optimal $\varepsilon$-robust program at distance $k - \eta$.
The optimization-added failure probability
(as in Corollary~\ref{cor:error})
for a run of (non-adaptive) \newalg completing depth $m$ and using this $\thresh$
is at most $e^{-2m(b + \eta)^2}$.
\end{claim}

%% file: eval.tex
\section{Evaluation}\label{sec:eval}

\paragraph{Implementation}
In this section, we evaluate our new algorithm \newalg (Figure~\ref{fig:newalg})
and its adaptive variant (Section~\ref{subsec:adaptive})
against  \alg (i.e., \newalg with $\thresh=1$). Both algorithms are implemented in Python
and use the SMT solver Z3~\cite{de2008z3} to implement
a sketch-based synthesizer $\osynth$.
We employ statistical verification for $\overif$ and $\oerror$:
we use  Hoeffding's inequality for estimating probabilities in $\post$ and $\error$.
Probabilities are computed
with 95\% confidence, leaving our oracles potentially unsound.

\paragraph{Research Questions}
Our evaluation aims to answer the following questions:
\begin{description}
\item[RQ1] Is adaptive \newalg more effective/precise than \newalg?
\item[RQ2] Is  \newalg more effective/precise  than \alg?
\item[RQ3] Can   \newalg solve challenging synthesis problems?
\end{description}

We experiment on three sets of benchmarks:
\rone synthetic examples for which the optimal solutions can be computed analytically~(Section~\ref{sec:toy}),
\rtwo the set of benchmarks considered in the original \alg paper~(Section~\ref{sec:original}),
\rthree a variant of the thermostat-controller synthesis problem presented in~\cite{chaudhuri14}~(Section~\ref{sec:controller}).
% In Section~\ref{sec:summary}, we summarize the answers to our research questions.
%All experiments were run on an ordinary consumer-grade computer.
%\sam{If that is too informal just delete it.}

\subsection{Synthetic Benchmarks}
\label{sec:toy}

We consider a class of synthetic programs for which we can compute the optimal solution exactly;
this lets us compare the results of our implementation to an ideal baseline.
Here, the program model \progs is defined as the set of axis-aligned hyperrectangles within $[-1,1]^d$ ($d \in \{1, 2, 3\}$
and the VC dimension is $2d$),
and the input distribution $\dist$ is such that inputs are
 distributed uniformly over $[-1,1]^d$.
We fix some probability mass $b \in \{0.05, 0.1, 0.2\}$ and
%to a reference implementation of the form
%$0 \leq x_1 \leq 2b \land \bigwedge_{i \in \{2, \ldots, d\}} -1 \leq x_i \leq 1$,
%which only returns 1 for points whose first coordinate is positive and at most $2b$.
%We fix the following postcondition:
%\[
%\pr{\prog(\vec{x}) = 1 \mid x_1 \leq 0} \geq \pr{\prog(\vec{x}) = 1 \mid x_1 \geq 0}
%\land
%\pr{\prog(\vec{x}) = 1} \geq b.
%\]
define the benchmarks  so  that the best error for a correct solution is exactly $b$
\iftoggle{full}{%
(see Appendix~\ref{app:apptoy}).
}{%
(for details, see~\cite{fullversion}).
}
%(and there exist dense regions of $\alpha$-robust programs
%that have error $b+\alpha$).

We run our implementation using thresholds $\thresh \in \{0.07, 0.15, 0.3, 0.5, 1\}$,
omitting those values for which $\thresh < b$;
additionally, we also consider an adaptive run
where $\thresh$ is initialized as the value $1$,
and whenever a new best solution is enumerated with error $k$,
we update $\thresh \gets k$.
Each combination of parameters was run for a period of 2 minutes.
Figure~\ref{fig:toy} fixates on $d=1$, $b=0.1$
and shows each of the following as a function of time:
\rone the depth completed by the search (i.e.\ the current size of the sample set), and
\rtwo the best solution found by the search.
\iftoggle{full}{%
(See Appendix~\ref{app:apptoy} for other configurations of $(d,b)$.)
}{%
(See our full version of the paper~\cite{fullversion} for other configurations of $(d,b)$.)
}

\begin{figure}[t]
\centering
\scriptsize
\pgfplotsset{filter discard warning=false}
\pgfplotscreateplotcyclelist{mylist}{
    {black, densely dotted},
    {black},
    {blue!50!black},
    {blue},
    {blue!50},
}
\begin{tikzpicture}
    \begin{groupplot}[
            group style={group size=2 by 1,
                horizontal sep=.5in},
            width=2.2in,
            cycle list name=mylist,
            ylabel near ticks,
            xlabel near ticks
        ]
        \nextgroupplot[
            xlabel=Time (s),
            ylabel=Depth Completed,
            xmin=0, xmax=120,
            ymin=0, ymax=130,
        ]
        \foreach \n in {a1, o1, o05, o03, o015}{
            \addplot+[const plot] table [x=\n_dvt_time, y=\n_dvt_depth, col sep=comma]
                {figures/toy_csv/box_d1_b02_s1.csv};
        }
        \nextgroupplot[
            legend to name={toylegend},
            legend style={/tikz/every even column/.append style={column sep=1em},
                legend columns=5},
            xlabel=Time (s) (log scale),
            ylabel=Best Error,
            xmin=0.3, xmax=120,
            ymin=0, ymax=0.5,
            xmode=log
        ]
        \foreach \n in {a1, o1, o05, o03, o015}{
            \addplot+[mark=diamond] table [x=\n_evt_time, y=\n_evt_error, col sep=comma]
                {figures/toy_csv/box_d1_b02_s1.csv};
        }
        \addplot[dashed, domain=0.3:120] {0.1};
        \addlegendentry{adaptive}
        \addlegendentry{$\thresh=1$}
        \addlegendentry{$\thresh=0.5$}
        \addlegendentry{$\thresh=0.3$}
        \addlegendentry{$\thresh=0.15$}
    \end{groupplot}
    \path (group c1r1.north east) -- node[above]{\ref{toylegend}} (group c2r1.north west);
\end{tikzpicture}
\caption{Synthetic hyperrectangle problem instance with parameters $d=1$, $b=0.1$.}
\label{fig:toy}
\end{figure}
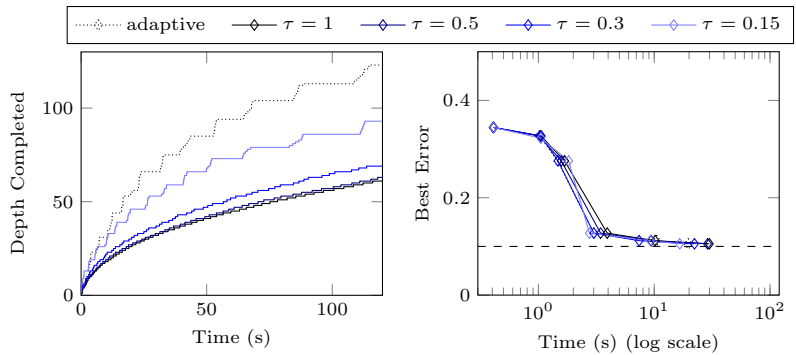

By studying Figure~\ref{fig:toy} we see that
the adaptive threshold search
performs at least as well as the tight thresholds fixed a priori
because reasonable solutions are found early.
In fact, all search configurations find solutions very close to the optimal error
(indicated by the horizontal dashed line).
%\loris{remove next sentence?}
%Interestingly, the fact that the small thresholds reach larger depths sooner
%does not seem to give them a significant temporal advantage in finding reasonable solutions,
%but this is part due to the simplicity of the examples.
Regardless, they reach different depths, and
\emph{the main advantage of reaching large depths
concerns the strength of the optimality guarantee.}
Note, also, that small $\thresh$ values are necessary to see
improvements in the completed depth of the search.
Indeed, the discrepancy between the depth-versus-time functions
diminishes drastically for the problem instances with larger values of $b$
\iftoggle{full}{%
(see Appendix~\ref{sec:apptoy});
}{%
(See our full version of the paper~\cite{fullversion});
}
the gains of the optimization are contingent on the existence
of correct solutions \emph{close} to the functional specification.

\textbf{Findings (RQ1):} \newalg \emph{does} tend to find \emph{reasonable} solutions at early depths
and near-optimal solutions at later depths, thus
adaptive \newalg is more effective than \newalg,
and we use it throughout our remaining experiments.

\subsection{Original DIGITS Benchmarks}
\label{sec:original}

The original \alg paper~\cite{digits} evaluates on a set
of 18 repair problems of varying complexity.
The functional specifications are machine-learned decision trees and support vector machines,
and each search space $\progs$ involves the set of programs formed
by replacing some number of real-valued constants in the program with holes.
The postcondition is a form of \emph{algorithmic fairness}---e.g., the program should output
true on inputs of type $A$ as often as it does on inputs of type $B$~\cite{feldman2015certifying}.
For each such repair problem,
we run both \alg
and adaptive \newalg (again, with initial $\thresh = 1$ and the identity refinement function).
Each benchmark is run for 10 minutes,
where the same sample set is used for both algorithms.

\begin{figure}[t]
\centering
\scriptsize
\begin{tikzpicture}
\begin{axis}[
        title=Depth Completed,
        xlabel=\alg,
        ylabel=adaptive \newalg,
        xmin=0, xmax=400,
        ymin=0, ymax=400,
        axis equal image, % enforces square shape
        width=.5\textwidth
    ]
    \addplot [only marks, mark=+]
        table [x=odepth, y=adepth, col sep=comma] {figures/comparison.csv};
    \addplot [dashed, domain=0:400] {x};
    \addplot [dotted, domain=0:400] {2.4*x};
\end{axis}
\end{tikzpicture}
\quad
\begin{tikzpicture}
\begin{axis}[
        title=Best Error,
        xlabel=\alg,
        ylabel=adaptive \newalg,
        xmin=0, xmax=.3,
        ymin=0, ymax=.3,
        axis equal image, % enforces square shape
        width=.5\textwidth
    ]
    \addplot [only marks, mark=+]
        table [x=oerror, y=aerror, col sep=comma] {figures/comparison.csv};
    \addplot [dashed, domain=0:.3] {x};
\end{axis}
\end{tikzpicture}
\caption{Improvement of using adaptive \newalg
on the original \alg benchmarks.
Left: the dotted line marks the $2.4\times$ average increase in depth.}
\label{fig:comparison}
\end{figure}
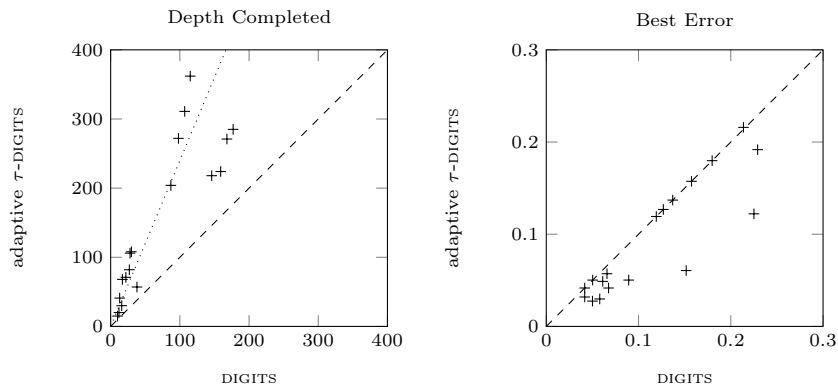

Figure~\ref{fig:comparison} shows, for each benchmark,
\rone the largest sample set size completed by adaptive \newalg versus \alg
(left---above the diagonal line indicates adaptive \newalg reaches further depths), and
\rtwo the error of the best solution found by adaptive \newalg versus \alg
(right---below the diagonal line indicates adaptive \newalg finds better solutions).
We see that adaptive \newalg reaches further depths on every problem instance,
many of which are substantial improvements,
and that it finds better solutions on 10 of the 18 problems.
For those which did not improve,
either the search was already deep enough that \alg was able to find near-optimal solutions,
or the complexity of the synthesis queries is such that the search
is still constrained to small depths.

\textbf{Findings (RQ2):} Adaptive \newalg
can find better solutions than those found by \alg and can reach greater search depths.

\subsection{Thermostat Controller}
\label{sec:controller}

We challenge adaptive \newalg with the task of synthesizing a thermostat controller,
borrowing the benchmark from~\cite{chaudhuri14}.
The input to the controller is the initial temperature of the environment;
since the world is uncertain, there is a specified probability distribution over the temperatures.
The controller itself is a program sketch consisting primarily of a single main loop:
iterations of the loop correspond to timesteps,
during which the synthesized parameters dictate an
incremental update made by the thermostat
based on the current temperature.
The loop runs for 40 iterations,
then terminates, returning the absolute value of the difference
between its final actual temperature and the target temperature.

The postcondition is a Boolean probabilistic correctness property
intuitively corresponding to controller safety,
e.g.\ with high probability, the temperature should never exceed certain thresholds.
In~\cite{chaudhuri14}, there is a quantitative objective in the form of
minimizing the expected value $\expec{|\mathit{actual}-\mathit{target}|}$---%
our setting does not admit optimizing with respect to expectations,
so we must modify the problem.
Instead, we fix some value $N$ ($N \in \{2, 4, 8\}$)
and have the program return $0$ when $|\mathit{actual}-\mathit{target}| < N$
and $1$ otherwise.
Our quantitative objective is to minimize the error from
the constant-zero functional specification $\initprog(x) \coloneqq 0$
(i.e.\ the actual temperature always gets close enough to the target).
\iftoggle{full}{%
The full specification of the controller is provided in Appendix~\ref{sec:apptherm}.
}{%
The full specification of the controller is provided in the full version of our paper~\cite{fullversion}.
}

%We only run  \newalg on the problem  because \alg cannot \loris{TODO}.
%There are two major details:
%first, the $\osynth$ queries are implemented by converting
%the program's input-output relation to a first order arithmetic formula,
%which is quite large due to the loop unrolling.
%Z3 does not scale to these formulae, and we instead use cvc4 \todo{cite this}
%which performed multiple orders of magnitude faster on this task.
%Second, even cvc4 is incredibly slow on the formulae,
%so we consider variants of the program where the thermostat runs for fewer timesteps,
%i.e.\ there are fewer loop iterations to unroll.

We consider variants of the program where the thermostat runs for fewer timesteps
and try loop unrollings of size $\{5, 10, 20, 40\}$.
We run each benchmark for 10 minutes:
the final completed search depths and best error of solutions
are shown in Figure~\ref{fig:therm}. For this particular experiment,
we use the SMT solver CVC4~\cite{cvc4} because it performs better than Z3 on the
occurring SMT instances.

\begin{figure}[t]
\centering
\scriptsize
\pgfplotsset{filter discard warning=false}
\pgfplotscreateplotcyclelist{mylist}{
    {black, mark=+},
    {blue!75!black, mark=x},
    {blue!50, mark=asterisk}
}
\pgfplotsset{ % taken from https://tex.stackexchange.com/questions/58548/
    discard if not/.style 2 args={
        x filter/.code={
            \edef\tempa{\thisrow{#1}}
            \edef\tempb{#2}
            \ifx\tempa\tempb
            \else
                \def\pgfmathresult{inf}
            \fi
        }
    }
}
\begin{tikzpicture}
    \begin{groupplot}[
            group style={group size=2 by 1,
                horizontal sep=.5in},
            width=2.2in,
            cycle list name=mylist,
            ylabel near ticks,
            xlabel near ticks,
            xlabel=Unrolling,
            xtick={5, 10, 20, 40},
            xticklabels={5, 10, 20, 40},
            xmode=log,
        ]
        \nextgroupplot[
            ylabel=Depth,
            ymin=0,
        ]
        \foreach \n in {8,4,2}{
            \addplot+[discard if not={bound}{\n}] table [x=unrolling, y=depth, col sep=comma]
                {figures/therm.csv};
        }
        \nextgroupplot[
            legend to name={thermlegend},
            legend style={/tikz/every even column/.append style={column sep=1em},
                legend columns=5},
            ylabel=Best Error,
            ymin=0, ymax=1
        ]
        \foreach \n in {8,4,2}{
            \addplot+[discard if not={bound}{\n}] table [x=unrolling, y=error, col sep=comma]
                {figures/therm.csv};
        }
        \addlegendentry{$N=8$}
        \addlegendentry{$N=4$}
        \addlegendentry{$N=2$}
    \end{groupplot}
    \path (group c1r1.north east) -- node[above]{\ref{thermlegend}} (group c2r1.north west);
\end{tikzpicture}
\caption{Thermostat controller results.}
\label{fig:therm}
\end{figure}
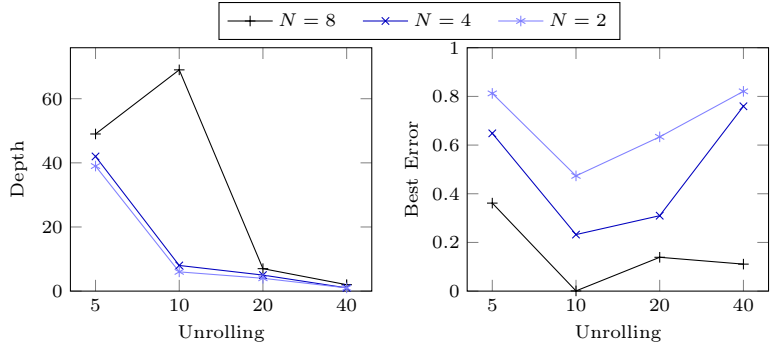

As we would expect, for larger values of $N$
it is ``easier'' for the thermostat to reach the target temperature threshold
and thus the quality of the best solution increases in $N$.
% \loris{simplify next sentence in a way that doesn't require knowing the benchmark in such detail}
However, with small unrollings (i.e.\ 5)
the synthesized controllers do not have enough iterations (time) to modify
the temperature  enough
for the probability mass of extremal temperatures to reach the target:
as we increase the number of unrollings to 10, we see that
better solutions can be found
since the set of programs are capable of stronger behavior.

On the other hand, the completed depth of the search plummets as the unrolling increases
due to the complexity of the $\osynth$ queries.
Consequently, for 20 and 40 unrollings,
adaptive \newalg synthesizes worse solutions
because it cannot reach the necessary depths to obtain better guarantees.

One final point of note is that
for $N=8$ and 10 unrollings, it seems that there is a sharp spike in the completed depth.
However, this is somewhat artificial:
because $N=8$ creates a very lenient quantitative objective,
an early $\osynth$ query happens to yield a program
with an error less than $10^{-3}$.
Adaptive \newalg then updates $\thresh \gets \approx 10^{-3}$
and skips most synthesis queries.
%any constraint strings mapping any number of inputs to 1.

\textbf{Findings (RQ3):} Adaptive \newalg can synthesize small variants of a complex thermostat controller,
but cannot solve variants with many loop iterations.
% \loris{everything after this point is too detailed and can be commented out}
% The only reason the completed depth does not exceed 1,000
% is because the active program is passed through solution propagation
% to the ``wrong'' child,
% and the algorithm hangs,
% performing the very expensive $\osynth$ query
% until timing out.
% Because $\lfloor \thresh \cdot \mathit{depth} \rfloor$ is so small,
% the failure bounds on missing a better program
% as presented in Corollary~\ref{cor:error}
% are quite large---%
% but this is \emph{because} the best synthesized program
% is so close to 0 error,
% and we have external confidence that
% it must already be practically optimal (with respective to additive factors).

% \subsection{Summary of Results}
% \label{sec:summary}
% Based on the experiments presented in the previous sections, we obtain the following answers to our research questions.
%
% \textbf{RQ1} Overall, the results indicate that \newalg
% can find better solutions than those found by \alg. \loris{expand}
%
% \textbf{RQ2} Overall, the results indicate that \newalg
% can feasibly reach bigger search depths than \alg. \loris{expand}
%
% \textbf{RQ3} \loris{summary stmt}

%% file: related.tex
\section{Related Work}\label{sec:related}

\paragraph{Synthesis \& Probability}
Program synthesis is a mature area with many powerful techniques.
The primary focus is on synthesis
under Boolean constraints, and probabilistic specifications have received less attention~\cite{digits,chaudhuri14,nori15,Kucera17}.
We discuss the works that are most related to ours.

\alg~\cite{digits} is the most relevant work. First, we show for the first time that
\alg only requires a number of synthesis queries  polynomial in the number of samples.
Second, our adaptive \newalg further reduces the number
of synthesis queries required to solve a synthesis problem without sacrificing correctness.

The technique of \emph{smoothed proof search}~\cite{chaudhuri14}
approximates a combination of functional correctness
and maximization of an expected value as a smooth, continuous function.
It then uses numerical methods to find a local optimum
of this function,
which translates to a synthesized program that is likely
to be correct and locally maximal.
The benchmarks described in Section~\ref{sec:controller} are variants of benchmarks
from~\cite{chaudhuri14}.
Smoothed proof search can minimize expectation; \newalg minimizes probability only.
However, unlike \newalg, smoothed proof search lacks
formal convergence guarantees and cannot support the rich
probabilistic postconditions we support, e.g., as in the fairness benchmarks.

Works on synthesis of probabilistic programs
are aimed at a different problem~\cite{nori15,chasins2017data,saad2019bayesian}: that of synthesizing a generative model of data. For example, Nori et al.~\cite{nori15} use sketches of probabilistic programs and complete them with a stochastic search.
Recently, Saad et al.~\cite{saad2019bayesian} synthesize an ensemble of probabilistic programs for learning Gaussian processes and other models.
% propose a stochastic technique for synthesizing a probabilistic program that is representative
% of a given input data set. This work is different from ours in two main aspects.
% First, Nori et al. are not interested in synthesizing program that meets a given  postcondition, but are interested in
% a program that ``approximates'' the input data.
% Second, while our technique provides strong convergence guarantess,
% the technique proposed in~\cite{nori15} only guarantees \emph{asymptotic} convergence in the limit because
% it is based on Montecarlo Markov Chain search.

K\v{u}cera et al.~\cite{Kucera17} present a technique for automatically synthesizing program transformations that
introduce uncertainty into a given program with the goal of satisfying given privacy policies---e.g., preventing information leaks.
They leverage the specific structure of their problem to reduce it to an SMT constraint solving problem.
The problem tackled in~\cite{Kucera17} is orthogonal to the one targeted in this paper and the techniques are therefore very different.

%In probabilistic model checking,
%a number of works have addressed the model
%repair problem, e.g.,~\cite{bartocci11,chen2013model}.
%In this line of work, the idea is to modify transition
%probabilities in finite-state Markov Decision Processes to
%satisfy a probabilistic temporal property.
%Our setting is quite different,
%in that we are modifying a program
%manipulating real-valued variables
%to satisfy a probabilistic postcondition.

\paragraph{Stochastic Satisfiability}
Our problem is closely related to \abr{E-MAJSAT}~\cite{littman1998computational},
a special case of \emph{stochastic satisfiability}
(\abr{SSAT})~\cite{papadimitriou1985games}
and a means for formalizing probabilistic planning
problems.
\abr{E-MAJSAT}
is of \abr{NP}$^\abr{PP}$ complexity.
An  \abr{E-MAJSAT} formula has deterministic
and probabilistic variables.
The goal is to find an assignment
of deterministic variables
such that the probability that the formula
is satisfied is above a given threshold.
Our setting is similar,
but we operate over complex
program statements and have an additional
optimization
objective (i.e., the program should be close to the functional specification).
The deterministic variables in our
setting are the holes defining the search space;
the probabilistic variables are program
inputs.

% There have been a number of investigations
% of algorithms for solving the problem.
% In the artificial intelligence community,
% Litmann et al.~\cite{littman2001stochastic}, for instance, proposed
% an extension of \abr{DPLL} to \abr{SSAT}
% (where there is an alternation of existential
% and probabilistic quantifiers).
% On the other hand,
% work from the verification community~\cite{franzle2008stochastic}
% proposed an extension of \abr{SSAT} to the
% theory of linear real arithmetic,
% but where all the probabilistic variables
% are Boolean.

% \paragraph{Computational Learning Theory}
% We have formalized correctness and
% convergence of \alg and \newalg in terms
% of \abr{PAC} learning and \abr{VC} dimensions~\cite{kearns1994introduction,michalski2013machine}.
% We are not the first to utilize computational
% learning theory in the context of synthesis and verification.
% For instance, a number of recent
% works~\cite{chen2016pac,sharma2013verification,krishna2015learning} use \abr{PAC} learning as a means for
% finding inductive invariants.
% \loris{add what you wish}

%% file: appex.tex
\section{Miscellaneous Proofs}

\subsection{Main Theorem}

\begin{proof}[Theorem~\ref{thm:main}]
Let $S$ be the set of samples, with $|S| = m$.
By Lemma~\ref{lem:algdich},
the number of queries is at most $|S| \lvert\dich{\progs}{S}\rvert$,
which is in turn at most $m \growth{\progs}{m}$.
Applying Lemma~\ref{lem:sauer} immediately gives us the $O(m^{d+1})$ bound.
\qed
\end{proof}

\subsection{Interval Details}

Here we expand on the details related to the set of interval programs
(Figure~\ref{fig:ex})
that were elided in the various examples in Section~\ref{sec:theory}.

\begin{claim}
For any (finite) set $S \subset [0,1]$,
$\lvert\dich{[0,a]}{S}\rvert = |S| + 1$.
\end{claim}
Technically, this claim is not correct:
when $0 \in S$, the number of dichotomies is one fewer.
However, $S$ is obtained by sampling from a distribution,
and if the distribution over $[0,1]$ does not contain atoms,
then this case almost surely does not happen.
We omit this detail for simpler presentation throughout.
\begin{proof}
Let the elements of $S = \{x_1,\ldots,x_m\}$ be ordered increasingly;
there are exactly $|S|+1$ equivalence classes of programs based on the choice of $a$:
one from $a < x_1$, one from $a > x_m$,
and $(|S|-1)$-many from $x_i < a < x_{i+1}$ for $i \in \{1,\ldots,m-1\}$.
\end{proof}

%% file: apptoy.tex
\section{Varying Synthetic Problem Parameters}\label{sec:apptoy}
\label{app:apptoy}

In this section, we provide the complete description of the synthetic benchmarks and present the complete plots
of our evaluation.

We consider a class of hyperrectangle programs for which we can compute the optimal solution exactly;
this lets us compare the results of our implementation to an ideal baseline.
Here, the concept class \progs (i.e., the set of programs) is defined as the set of axis-aligned hyperrectangles within $[-1,1]^d$,
and the input distribution $\dist$ is such that inputs are
 distributed uniformly over $[-1,1]^d$.
We fix some probability mass $b$ and aim to synthesize a program that is close
to a functional specification of the form
$0 \leq x_1 \leq 2b \land \bigwedge_{i \in \{2, \ldots, d\}} -1 \leq x_i \leq 1$,
which only returns 1 for points whose first coordinate is positive and at most $2b$.
We fix the following postcondition:
\[
\pr{\prog(\vec{x}) = 1 \mid x_1 \leq 0} \geq \pr{\prog(\vec{x}) = 1 \mid x_1 \geq 0}
\land
\pr{\prog(\vec{x}) = 1} \geq b.
\]
In other words, a correct hyperrectangle must include as much probability mass
of points whose first coordinate is negative as it does for those with a positive first coordinate,
and additionally it must include at least as much probability mass as the original hyperrectangle.
Observe that independent of $d$, the best error for a correct solution is exactly $b$
(and there exist dense regions of $\alpha$-robust programs
that have error $b+\alpha$).

We consider problem instances formed from combinations of $d \in \{1, 2, 3\}$
and $b \in \{0.05, 0.1, 0.2\}$.
As $d$ increases, the set of programs increases in complexity
(in fact, it has VC dimension $2d$)
and the synthesis queries become more expensive.
As $b$ increases, the threshold used by the optimization cannot be as small,
so we expect the search to benefit less from our optimizations.
We run our implementation using thresholds $\thresh \in \{0.07, 0.15, 0.3, 0.5, 1\}$,
omitting those values for which $\thresh < b$;
additionally, we also consider an adaptive run
where $\thresh$ is initialized as the value $1$,
and whenever a new best solution is enumerated with error $k$
we update $\thresh \gets k$.

Each combination of parameters was run for a period of 2 minutes.
Figure~\ref{fig:toygraphs} shows
each of the following as a function of time:
\rone the depth completed by the search (i.e.\ the current size of the sample set), and
\rtwo the best solution found by the search.

\begin{figure}[t]
\centering

\pgfplotsset{filter discard warning=false}
\pgfplotscreateplotcyclelist{mylist}{
    {black, densely dotted},
    {black},
    {blue!40!black},
    {blue!80!black},
    {blue!80},
    {blue!40}
}
\newcommand{\depthplot}[2]{%
\begin{tikzpicture}[baseline]
    \begin{axis}[
            tiny,
            cycle list name=mylist,
            xmin=0, xmax=120,
            ymin=0, ymax=150,
            width=1.8in,
            height=1.5in
        ]
        \foreach \n in {#2}{
            \addplot+[const plot, no marks] table
                [x=\n_dvt_time, y=\n_dvt_depth, col sep=comma]{#1};
        }
    \end{axis}    
\end{tikzpicture}
}
\newcommand{\errorplot}[4][]{%
\begin{tikzpicture}[baseline]
    \begin{axis}[
            tiny,
            cycle list name=mylist,
            xmin=0.3, xmax=120,
            ymin=0, ymax=0.5,
            xmode=log,
            width=1.8in,
            height=1.5in,
            legend style={row sep=-2pt}
        ]
        \foreach \n in {#3}{
            \addplot+[mark=diamond] table [x=\n_evt_time, y=\n_evt_error, col sep=comma]{#2};
        }
        \addplot[dashed, domain=0.3:120] {#4};
        \ifstrempty{#1} {%
            % no legend
        } {%
            \legend{adaptive, $\thresh=1$, $\thresh=0.5$,
                    $\thresh=0.3$, $\thresh=0.15$, $\thresh=0.07$}
        }
    \end{axis}
\end{tikzpicture}
}

Plots of Completed Depth vs Time (s)

{\scriptsize
\begin{tabular}{crrr}
& \multicolumn{1}{c}{dim = 1} & \multicolumn{1}{c}{dim = 2} & \multicolumn{1}{c}{dim = 3} \\
\rotatebox{90}{~optimal at 0.05} &
\depthplot{figures/toy_csv/box_d1_b01_s1.csv}{a1, o1, o05, o03, o015, o007} &
\depthplot{figures/toy_csv/box_d2_b01_s1.csv}{a1, o1, o05, o03, o015, o007} &
\depthplot{figures/toy_csv/box_d3_b01_s1.csv}{a1, o1, o05, o03, o015, o007} \\
\rotatebox{90}{~optimal at 0.1} &
\depthplot{figures/toy_csv/box_d1_b02_s1.csv}{a1, o1, o05, o03, o015} &
\depthplot{figures/toy_csv/box_d2_b02_s1.csv}{a1, o1, o05, o03, o015} &
\depthplot{figures/toy_csv/box_d3_b02_s1.csv}{a1, o1, o05, o03, o015} \\
\rotatebox{90}{~optimal at 0.2} &
\depthplot{figures/toy_csv/box_d1_b04_s1.csv}{a1, o1, o05, o03} &
\depthplot{figures/toy_csv/box_d2_b04_s1.csv}{a1, o1, o05, o03} &
\depthplot{figures/toy_csv/box_d3_b04_s1.csv}{a1, o1, o05, o03}
\end{tabular}
}

Plots of Best Error vs Time (s)

{\scriptsize
\begin{tabular}{crrr}
& \multicolumn{1}{c}{dim = 1} & \multicolumn{1}{c}{dim = 2} & \multicolumn{1}{c}{dim = 3} \\
\rotatebox{90}{~optimal at 0.05} &
\errorplot[legend]{figures/toy_csv/box_d1_b01_s1.csv}{a1, o1, o05, o03, o015, o007}{.05} &
\errorplot{figures/toy_csv/box_d2_b01_s1.csv}{a1, o1, o05, o03, o015, o007}{.05} &
\errorplot{figures/toy_csv/box_d3_b01_s1.csv}{a1, o1, o05, o03, o015, o007}{.05} \\
\rotatebox{90}{~optimal at 0.1} &
\errorplot{figures/toy_csv/box_d1_b02_s1.csv}{a1, o1, o05, o03, o015}{.1} &
\errorplot{figures/toy_csv/box_d2_b02_s1.csv}{a1, o1, o05, o03, o015}{.1} &
\errorplot{figures/toy_csv/box_d3_b02_s1.csv}{a1, o1, o05, o03, o015}{.1} \\
\rotatebox{90}{~optimal at 0.2} &
\errorplot{figures/toy_csv/box_d1_b04_s1.csv}{a1, o1, o05, o03}{.2} &
\errorplot{figures/toy_csv/box_d2_b04_s1.csv}{a1, o1, o05, o03}{.2} &
\errorplot{figures/toy_csv/box_d3_b04_s1.csv}{a1, o1, o05, o03}{.2}
\end{tabular}
}

\caption{Performance on synthetic hyperrectangle examples with varying parameters.}
\label{fig:toygraphs}
\end{figure}
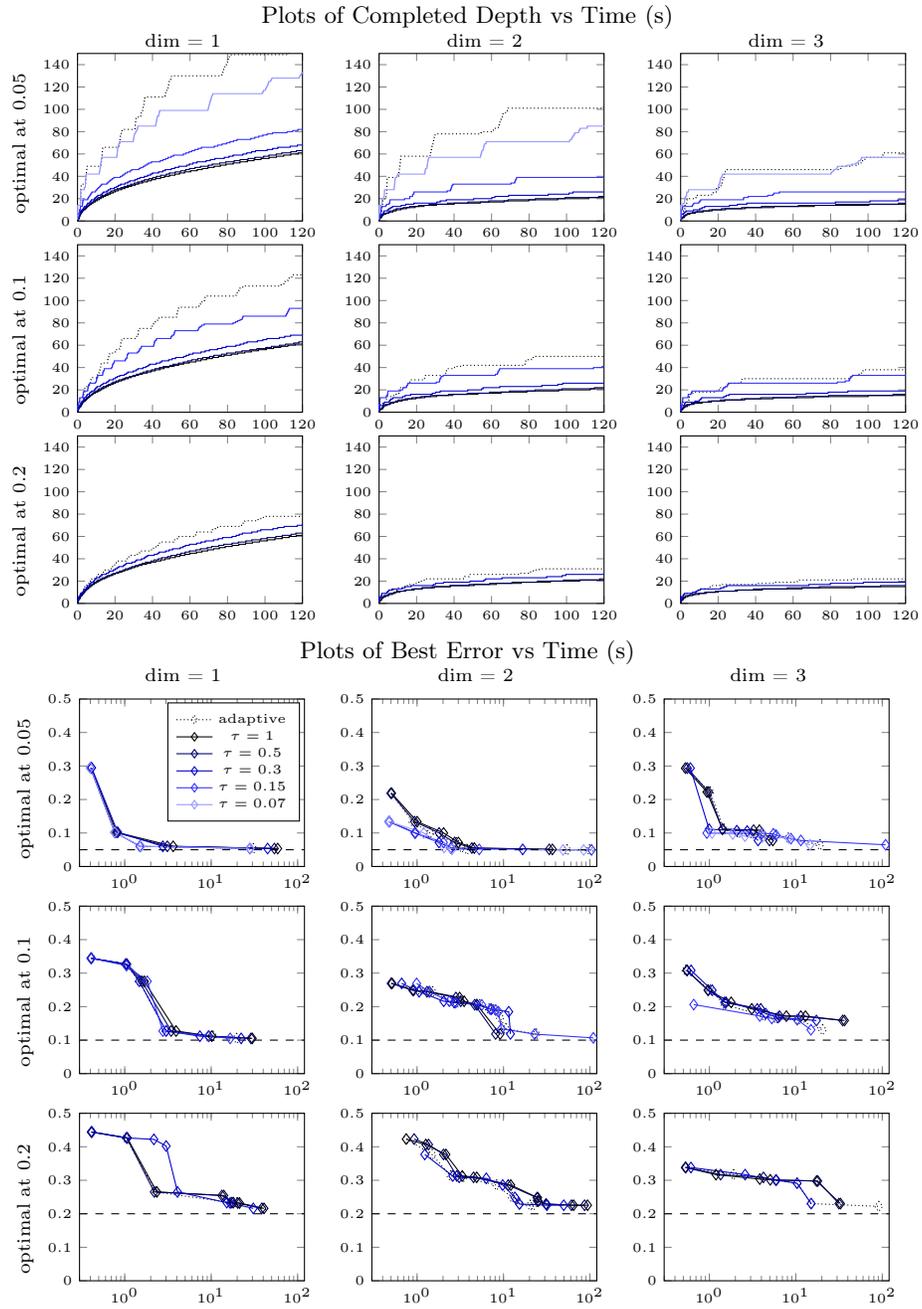

%% file: apptherm.tex
\section{Thermostat Benchmark}\label{sec:apptherm}

Here we include the specification of our modified version
of the thermostat controller synthesis benchmark~\cite{chaudhuri14}.
Figure~\ref{fig:thermcode} shows the definitions of
\texttt{pre}, which describes the probability distribution $\dist$ over the inputs,
and \texttt{thermostat},
a program sketch describing the set of possible programs.
We handle the thermostat loop (line~\ref{line:thermloop})
through syntactic unrolling, since it has a constant bound:
\emph{Unrollings} is the value we instantiate from $\{5, 10, 20, 40\}$
in the creating of problem instances for our experiments.
Similarly, the threshold \emph{N} in line~\ref{line:thermN}
is instantiated from $\{2, 4, 8\}$.

\begin{figure}[t]
\centering
\lstset{
    language=C++,
    basicstyle=\fontfamily{pcr}\selectfont\scriptsize,
    alsoletter=?,
    morekeywords={assert,??},
    numbers=left,
    escapechar=|,
    tabsize=2,
    literate={\ \ }{{\ }}1 % manually replace 4-space tabbing with 2-space
}
\begin{tabular}{l}
\begin{lstlisting}
double, double pre() {
    double modal = Uniform({1, 2, 3});
    double lin;
    if(modal == 1) {
        lin = gaussian(mean=30, variance=9);
    } else if (modal == 2) {
        lin = gaussian(mean=35, variance=9);
    } else {
        lin = gaussian(mean=50, variance=9);
    }
    double ltarget = gaussian(mean=75, variance=1);
    return lin, ltarget;
}
\end{lstlisting}
\vspace{1em}
\\
\begin{lstlisting}[
        emph={N,Unrollings},
        emphstyle={\itshape}
    ]
int thermostat(double lin, double ltarget) {
    double h = ??(0,10);|\label{line:thermhole}|
    double tOn = ltarget + ??(-10,0);
    double tOff = ltarget + ??(0,10);
    double isOn = 0.0;
    double K = 0.1;
    double CurL = lin;
    assert(tOn < tOff; 0.9);|\label{line:thermassert}|
    assert(h > 0; 0.9);
    assert(h < 20; 0.9);
    for(int i = 0; i < Unrollings; i = i + 1) {|\label{line:thermloop}|
        if(isOn > 0.5) {
            curL = curL + (h - K * (curL - lin));
            if(curL > tOff) {
                isOn = 0.0;
            }
        } else {
            curL = curL - K * (curL - lin);
            if(curL < tOn) {
                isOn = 1.0;
            }
        }
        assert(curL < 120; 0.9);
    }
    Error = abs(curL - ltarget);
    if(Error < N) {|\label{line:thermN}|
        return 0;
    } else {
        return 1;
    }
}
\end{lstlisting}
\end{tabular}
\caption{Program sketch defining the set of possible thermostat controllers.}
\label{fig:thermcode}
\end{figure}

A synthesized program instantiates the sketch
by replacing the \emph{holes} with real-valued constants:
for example, the syntax in the thermostat definition at line~\ref{line:thermhole}
specifies that the synthesizer must replace the right side of the assignment
with a constant between 0 and 10.
The \texttt{assert} statements form the probabilistic postcondition:
if we have the set of assert statements in the program 
$\{\texttt{assert}(\mathit{event}_i; \theta_i);\}_{i \in I}$,
then the postcondition is given by the following conjunction:
$\bigwedge_{i \in I} \pr{\mathit{event}_i} > \theta_i$.
(Recall that the loop is syntactically unrolled
and observe that all execution paths encounter all \texttt{assert} statements,
so this is well-defined.)